\documentclass[a4paper,11pt]{article}

\usepackage[utf8]{inputenc}
\usepackage[english]{babel}
\usepackage{xspace}
\usepackage{pdfsync}

\usepackage{theoryofcomputing} 

\usepackage[bookmarks=true, bookmarksnumbered=true, pagebackref=true, colorlinks=true, linkcolor=blue, citecolor=red]{hyperref}

\title{The complexity of proving that a graph is Ramsey}
\author{%
  Massimo Lauria\\
  \texttt{lauria@kth.se}\\
  Royal Institute of Technology, Stockholm
  \and
  Pavel \Pudlak\\
  \texttt{pudlak@cas.cz}\\
  Czech Academy of Sciences, Prague
  \and
  \Vojtech \Rodl\\
  \texttt{rodl@mathcs.emory.edu}\\
  Emory University, Atlanta
  \and
  Neil Thapen\\
  \texttt{thapen@cas.cz}\\
  Czech Academy of Sciences, Prague
}
\date{\today}

\AtBeginDocument{\hypersetup{%
  pdftitle = {The complexity of proving that a graph is Ramsey},
  pdfauthor = {Massimo Lauria, Pavel Pudlak, Vojtech Rodl and Neil Thapen}
}}

\newcommand{\ignore}[1]{}
\newcommand{\clique}{\mathrm{Clique}}
\newcommand{\NB}{N}
\newcommand{\ram}{\Psi}

\begin{document} 


\maketitle

\begin{abstract}
We say that a graph with $n$ vertices is $c$-Ramsey if it does not contain either a clique
or an  independent set of size $c \log n$.  We define a CNF formula
which expresses this property for a graph $G$.
We show a superpolynomial lower bound
on the length of resolution proofs  that $G$ is $c$-Ramsey,  for
\emph{every} graph $G$.
Our proof makes use of the fact that every Ramsey graph must contain
a large subgraph with some of the statistical properties of the random graph.
\end{abstract}


\section*{Introduction}
Graphs with special properties often require non trivial and/or probabilistic constructions.
Furthermore, once the graph is constructed it may be hard to verify that the property holds, and if
such a graph is given to a new user without a suitable certificate, he must either
 verify the construction again or blindly trust the graph.

In this paper we are interested in how hard it is to certify that a graph $G$ of size $n$ is $c$-Ramsey,
that is, has no clique or independent set of length $c \log n$.
Constructing such graphs was one of the first applications of the probabilistic method
in combinatorics. But the brute force approach to checking that $G$ satisfies
the property takes time $n^{O(\log n)}$
(compare the well-known hard problem of looking for cliques).
We show that there is no resolution proof that $G$ is $c$-Ramsey with length shorter than
$n^{O(\log n)}$. This is not a worst-case result, but rather holds for \emph{every} graph $G$.
However we are only able to show this for what we call the ``binary'' formalization
of the Ramsey property as a propositional formula; for an alternative,
``unary'' formalization we only know a treelike resolution lower bound
(see Section \ref{sec:openproblems}).

Notice that the lower bound on resolution proof size shows that the verification
problem is hard for quite a large class of algorithms, since most SAT
solvers used in practice are essentially proof search algorithms for resolution~\cite{Pipatsrisawat2011512}.
Notice also that, while it does not follow from the resolution
lower bound that there is no algorithm which will construct a Ramsey graph
in polynomial time, it does follow that, given such an algorithm, there is  no
polynomial-size resolution proof that the algorithm works.

The finite Ramsey theorem states that for any $k$, there is some $N$ such that every
graph of size at least $N$ contains a clique or independent set of size $k$. We write $r(k)$
for the least such $N$.
Computing the actual value of $r(k)$ is challenging, and so far only a few values
have been discovered.  For this reason there is great interest in asymptotic estimates~\cite{erdos1947some,spencer77,conlon2009}.

A $c$-Ramsey graph is a witness that $r(c \log n) > n$, so proving that a graph
is Ramsey is in some sense proving a lower bound for $r(k)$.
Previously, proof complexity has  focused on upper bounds for $r(k)$.
Krishnamurthy and Moll~\cite{km81} proved partial results on the complexity of proving the exact upper bound, and conjectured this formula to be hard in general. \Krajicek later proved an exponential lower bound on the length of bounded depth Frege proofs of the same statement~\cite{kra10}.
The upper bound $r(k) \le 4^k$ has short proofs in a relatively weak fragment of sequent calculus,
in which every formula in a proof has small constant depth~\cite{pudlak91}, \cite{kra10}.
Recently \Pudlak \cite{pudlak2012ramsey} has shown a lower bound on proofs of $r(k) \le 4^k$
in resolution. We discuss this in more detail in Section 1.
There are also results known about the off-diagonal Ramsey numbers $r(k,s)$ where cliques of size $k$ and independent sets of size $s$ are considered. See~\cite{es35,aks80,kim95,bohman2010early} for estimates and~\cite{cgl10} for resolution lower bounds.

\medskip

In Section 1 we formally state our main result, mention some open problems,
and then outline the high-level method we will use. In Section 2 we apply this
to prove a simple version of our main theorem, restricted to the case when $G$ is
a random graph. In Section 3 we prove the full version. This will use one extra ingredient,
a result from \cite{promel1999non} that every Ramsey graph $G$ has a large
subset with some of the statistical density properties of the random graph.

\section{Definitions and results}


Resolution~\cite{Blake1937} is a system for refuting propositional
CNFs, that is, propositional formulas in conjunctive normal form. A
resolution refutation is a sequence of disjunctions, which
in this context we call \emph{clauses}. Resolution has a
single inference rule: from two clauses $A \lor x$ and $B \lor \neg x$
we can infer the new clause $A \lor B$ (which is a logical
consequence).  A \emph{resolution refutation} of a CNF $\phi$ is a derivation
of the empty clause from the clauses of $\phi$.
For an unsatisfiable formula $\phi$ we define $L(\phi)$ to be the length,
that is, the number of clauses, of the shortest resolution refutation of $\phi$.
If $\phi$ is satisfiable we consider $L(\phi)$ to be infinite.

Let $c>0$ be a constant, whose value will be fixed for the rest of the paper.

\begin{definition}[Ramsey graph]
  We say that a graph with $n$ vertices is \mbox{\emph{$c$-Ramsey}} if there is no set of $c \log n$ vertices
	which form either a clique or an independent set.
\end{definition}

We now describe how we formalize this in a way suitable for the resolution
proof system.
Given a graph $G$ on $n=2^k$ vertices,
we will define a formula $\ram_G$ in conjunctive normal form
which is satisfiable if and only there is a homogeneous set of size $ck$ in $G$,
that is, if and only if $G$ is not Ramsey.
We identify the vertices of $G$ with the binary strings of length $k$. In this
way we can use an assignment to $k$ propositional variables to determine a vertex.

The formula $\ram_G$ has variables to represent an injective mapping from a set of
$ck$ ``indices'' to the vertices of $G$, and asserts that the vertices mapped
to form either a clique or an independent set. It has a single extra
variable $y$ to indicate which of these two cases holds.

In more detail, for each $i \in [ck]$ we have $k$ variables $x^i_1, \dots, x^i_k$
which we think of as naming, in binary, the vertex of $G$ mapped to by $i$.
We have an additional variable $y$, so there are $ck^2+1$ variables in total.
To simplify notation we will write propositional literals in the form
``$x^i_b=1$'', ``$x^i_b\not=0$'', ``$x^i_b=0$'' and ``$x^i_b\not=1$''.
The first and the second are aliases for the literal $x^i_b$.
The third and the fourth are aliases for literal $\neg x^i_b$.

The formula $\ram_G$ then consists of clauses asserting the following:

\begin{enumerate}
\item
\textbf{The map is injective.}
For each vertex $v\in V(G)$, represented as $v_{1} \cdots v_{k}$ in binary,
and each pair of distinct $i,j \in [ck]$,  we have the clause
\[ \label{eq:clause_injective}
\bigvee_{b=1}^{k} (x^i_b \not= v_b) \lor \bigvee_{b=1}^{k} (x^j_b \not= v_b).
\]
These clauses guarantee that no two indices $i$ and $j$
map to the same vertex~$v$.

\item
\textbf{If $y=0$, then the image of the mapping is an independent set.}
For each pair of distinct vertices $u,v \in V(G)$, represented respectively
as $u_1 \dots u_k$ and $v_1 \dots v_k$, and each pair of distinct $i,j \in [ck]$,
if $\{ u, v \} \in E(G)$ we have the clause
\[
y \vee \bigvee_{y=1}^{k} (x^i_b \not= u_b) \lor \bigvee_{b=1}^{k} (x^j_b \not= v_b).
\]
These clauses guarantee that, if $y=0$, then no two indices are mapped
to two vertices with an edge between them.

\item
\textbf{If $y=1$, then the image of the mapping is a clique.}
For each pair of distinct vertices $u,v \in V(G)$, represented respectively
as $u_1 \dots u_k$ and $v_1 \dots v_k$, and each pair of distinct $i,j \in [ck]$,
if $\{ u, v \} \notin E(G)$ we have the clause
\[
\neg y \vee \bigvee_{b=1}^{k} (x^i_b \not= u_b) \lor \bigvee_{b=1}^{k} (x^j_b \not= v_b).
\]
These clauses guarantee that, if $y=1$, then no two indices are mapped
to two vertices without an edge between them.
\end{enumerate}
Notice that the formula has $\binom{ ck}{2 }\left( 1+ \binom{n}{2} \right)$ clauses
in total,
and so is unusual in that the number of clauses is exponentially
larger than the number of variables.
However the number of clauses is polynomial in the number $n$ of vertices of $G$.

If $G$ is Ramsey, then $\ram_G$ is unsatisfiable and only has
$c \log^2 n +1$ variables. So we can refute $\ram_G$ in quasipolynomial
size by a brute-force search through all assignments:

\begin{proposition} \label{pro:bruteforceupperbound}
If $G$ is $c$-Ramsey, the formula $\ram_G$ has a (treelike) resolution refutation of size
$n^{O(\log n)}$.
\end{proposition}

At this point, we should recall the formalization of the Ramsey theorem that
is more usually studied in proof complexity.
This is the family $\mathrm{RAM}_n$
of propositional CNFs, where $\mathrm{RAM}_n$ has one variable for each
distinct pair of points in $[n]$ and asserts that the graph represented
by these variables is $\frac{1}{2}$-Ramsey.
Hence $\mathrm{RAM}_n$ is satisfiable if and only if any $\frac{1}{2}$-Ramsey graph
on $n$ vertices exists.
In contrast, our formula $\ram_G$ is satisfiable
if and only if our particular graph $G$ is not $c$-Ramsey.

Put differently, a refutation of $\mathrm{RAM}_n$ is a proof that
$r(k) \le 2^{2k}$. This was recently shown to require exponential size
(in $n$) resolution refutations \cite{pudlak2012ramsey}.
On the other hand a refutation of $\ram_G$ is a proof that $G$ is $c$-Ramsey, and hence
that $G$ witnesses that $r(k) > 2^{\frac{k}{c}}$.

We now state our main result. We postpone the proof to Section 3.

\begin{theorem} \label{the:ramseylowerbound}
Let $G$ be any graph with $n$ vertices. Then $L (\ram_G)  \ge n^{\Omega( \log n )}$.
\end{theorem}

If $G$ is not $c$-Ramsey then this is trivial, since $\ram_G$ is satisfiable
and therefore $L(\ram_G)$ is infinite by convention.
If $G$ is $c$-Ramsey, then by Proposition
\ref{pro:bruteforceupperbound} this bound is tight
and we know that  $L(\ram_G) = n^{\Theta(\log n)}$.

\subsection{Open problems} \label{sec:openproblems}

A shortcoming of our result is that our formula $\ram_G$
asserting that a graph is not $c$-Ramsey identifies the vertices
of $G$ with binary strings. It could be argued that this ``binary encoding''
of the statement brings some extra structure to the graph,
and that a formalization which does not do this is more
combinatorially natural.

So consider the ``unary encoding'' $\ram'_G$,
in which  the mapping from an index~$i$ to the vertices
of $G$ is represented by $n$ variables $\{p^i_v :  v \in V(G) \}$
and we have clauses asserting that for each $i$,
exactly one of the variables $p^i_v$ is true.
Otherwise the structure of $\ram'_G$ is similar to that of $\ram_G$.
As before, if $G$ is a $c$-Ramsey graph we have the brute-force upper bound
$L(\ram'_G) = n^{O(\log n)}$. But we are not able to prove
a superpolynomial lower bound on resolution size.
However if we restrict to treelike resolution, such a lower bound
follows using techniques from~\cite{bgl11paradpll}.  Here we are able
to prove the tree-like resolution lower bound as a corollary of our
main theorem (we are grateful to Leszek Ko{\l}odziejczyk for pointing
out this simpler proof).

\begin{theorem}
Let $G$ be any $c$-Ramsey graph with $n$ vertices. Then $\ram'_G$ requires
treelike resolution refutations of size $n^{\Omega(\log n)}$.
\end{theorem}

\def\Res{\mathrm{Res}}

\begin{proof} (Sketch)
Suppose we have a small treelike resolution refutation of the unary formula
$\ram'_G$. We can produce from it an at most polynomially larger treelike
$\Res(k)$ refutation of the binary formula $\ram_G$ as follows.
Replace each variable $p^i_v$ asserting that index $i$ is mapped to vertex $v$
with the conjunction $\bigwedge_{b=1}^k x^i_b = v_b$.
The substitution instance of $\ram'_G$ is then almost identical to the
$\ram_G$, except for the additional clauses asserting that every index
maps to exactly one vertex; but these are easy to derive in
 treelike $\Res(k)$.

It is well-known that every treelike depth $d+1$ Frege proof can be made
into a daglike depth $d$ Frege proof with at most polynomial increase in size \cite{kra1994lower}.
In particular, we can turn our treelike $\Res(k)$ refutation of
$\ram_G$ into a resolution refutation. The lower bound then follows from
Theorem \ref{the:ramseylowerbound}.
\end{proof}

Lower bounds for daglike resolution would have interesting
consequences for various area of proof
complexity~\cite{aft11,dms2011j}. This is related to the following open problem
(rephrased from \cite{bglr12}):
  consider a random graph $G$ distributed according to $\mathcal{G}(n,n^{-(1+\epsilon)\frac{2}{k-1}})$ for some $\epsilon>0$. Does every resolution proof that there is no $k$-clique in $G$
 require size $n^{\Omega(k)}$? For tree-like resolution this problem
 has been solved in~\cite{bgl11paradpll}.

\subsection{Resolution width and combinatorial games}
\label{sec:game}

The \emph{width} of a clause is the number of literals it contains.
The width of a CNF $\phi$ is the width of its widest clause. 
Similarly the width of a resolution refutation $\Pi$ is the width of its widest clause.
The width of refuting an unsatisfiable CNF $\phi$
is the minimum width of $\Pi$ over all refutations $\Pi$ of $\phi$.
We will denote it by~$W(\phi)$.

A remarkable result about resolution is that it is possible to lower bound the proof length by lower bounding the proof width.

\begin{theorem}[\cite{bw99}]\label{thm:size_width}
  For any CNF $\phi$ with $m$ variables and width $k$,
\[
  L(\phi) \geq 2^{\Omega \left(
 \tfrac{{(W(\phi)-k)}^2}{m}
 \right)}.
\]
\end{theorem}

\ignore{
There are several combinatorial games that model different complexity measures for resolution~\cite{et99,et03,bg03b,ad08,pudlak00,bgl10c}.
 In particular it is possible to characterize resolution width by
the \emph{Boolean existential pebble game}~\cite{ad08}.
[** I don't see the need for this paragraph, but I'm not sure
what name or attribution to give to the following game and theorem.
Calling it the Boolean existential pebble game certainly seems too much.
I am tempted to call the theorem folklore, but it may be that
Atserias-Dalmau is the right reference for the precise result in both
directions.  - Neil **]
}

Now consider a game played between two players, called the Prover and the Adversary.
The Prover claims that a CNF $\phi$ is unsatifiable and
the Adversary claims to know a satisfying assignment.
At each round of the game the Prover asks for the value of some variable
and the Adversary has to answer.  The Prover saves the answer in memory,
where each variable value occupies one memory location.
The Prover can also delete any saved value, in order to save memory.
If the deleted variable is asked again, the Adversary is allowed to
answer differently.
The Prover wins when the partial assignment in memory falsifies a clause of $\phi$.
The Adversary wins if he has a strategy to play forever.

If $\phi$ is in fact unsatisfiable, then the Prover can always eventually win,
by asking for the total assignment. If $\phi$ is satisfiable, then
there is an obvious winning  strategy for the Adversary (answering according to a fixed satisfying assignment).
However, even if $\phi$ is unsatisfiable,
it may be that the Prover cannot win the game unless he uses a large amount of memory.
Indeed, it turns out that smallest number of memory locations that the Prover needs to win the game for an unsatisfiable $\phi$ is related to the width of resolution refutations.
(We only need one direction of this relationship -- for a converse see \cite{ad08}.)

\begin{lemma}\label{lem:game_width}
  Given an unsatisfiable CNF $\phi$, it holds that $W(\phi)+1$ memory locations are sufficient for the Prover in order to win the game against any Adversary.
\end{lemma}

\subsection{The clique formula}\label{sec:cliqueformula}

For any graph $G$, the formula $\ram_G\!\!\restriction_{y=1}$
is satisfiable if and only if $G$ has a clique of size $ck$.
We will call this restricted formula $\clique(G)$.
Dually, $\ram_G\!\!\restriction_{y=0}$ is equivalent to $\clique(\bar{G})$.
Since fixing a variable in a resolution refutation results in a refutation for the
corresponding restricted formula, we have
\begin{equation*} \label{eq:reduction_from_clique}
  \max\left\{L(\clique(G)),L(\clique(\bar{G}))\right\} \leq  L(\ram_G) .
\end{equation*}
Furthermore we can easily construct a refutation of $\ram_{G}$ from
refutations of $\ram_{G}\!\!\restriction_{y=1}$ and $\ram_{G}\!\! \restriction_{y=0}$.
In this way we get
\begin{equation*}
L(\ram_{G})  \leq L(\clique(\bar{G}))+L(\clique(G)) + 1.
\end{equation*}

We can now describe our high-level approach.
To lower-bound $L(\ram_G)$ it is enough to lower-bound $L(\clique(G))$,
which we will do indirectly by exhibiting a good strategy for the Adversary
in the game on $\clique(G)$.
This game works as follows: the Adversary claims
to know $ck$ strings in $\{ 0,1 \}^k$ which
name $ck$ vertices in $G$ which form a clique.
The Prover starts with no knowledge of these strings
but can query them, one bit at a time,
and can also forget bits to save memory.
The Prover wins if at any point there are two
fully-specified strings for which the corresponding vertices
 are not connected by an edge in $G$.

We will give a strategy for the Adversary  which will
beat any Prover limited to $\epsilon k^2$ memory for a constant $\epsilon>0$.
It follows by Lemma~\ref{lem:game_width} that $\clique(G)$
is not refutable in width $\epsilon k^2$.
The formula $\clique(G)$ has $ck^2$ variables and has
width $2k$. Hence applying Theorem~\ref{thm:size_width}
we get
\[
L(\ram_G)
\geq
L(\clique(G))
 \geq
 2^{\Omega \left(
 \tfrac{(\epsilon k^2 - 2k)^2}{ck^2}
 \right)}
\geq
2^{\Omega(k^2)}
\geq
n^{\Omega( \log n )}.
\]

\subsection{Other notation}
We will consider simple graphs with $n=2^{k}$ vertices. We
identify the vertices with the binary strings of length $k$.
For any vertex $v \in G$ we denote its binary representation by
 $ v_{1}\cdots v_{k} $.

A \emph{pattern} is a partial assignment to $k$ variables.
Formally, it is
 a string $p = p_{1} \cdots p_{k}  \in {\{*,0,1\}}^{k}$,
and we say that $p$ \emph{is consistent with} $v$ if for all $i\in[k]$ either
$p_{i}=v_{i}$ or $p_i = *$.
The \emph{size} $|p|$ of $p$ is the number of bits set to $0$ or $1$.
The \emph{empty pattern} is a string of $k$ stars.

For any vertex $v \in V(G)$ we let $\NB(v)$ be the set
$\{u \big| \{v,u\} \in E(G)\}$ of neighbours of $v$. Notice that $v\not\in\NB(v)$.
For any $U\subseteq V(G)$ we let $\NB(U)$ be the set of vertices of $G$ which
neighbour every point in $U$, that is, $\bigcap_{v \in U} \NB(U)$.
Notice that $U \cap \NB(U)=\emptyset$.

\section{Lower bounds for the random graph}\label{sec:random-graphs}

We consider random graphs on $n$ vertices given by
the usual distribution
$\mathcal{G}(n,\frac{1}{2})$ in the Erd\H{o}s-R\'{e}nyi model.

\begin{theorem}
If $G$ is a random graph, then with high probability
$L(\ram_G) = n^{ \Omega( \log n )}$.
\end{theorem}

We will use the method outlined in Section~\ref{sec:cliqueformula}
above, so to prove the theorem it is enough to give a strategy for the Adversary
in the game on
$\clique(G)$ which forces the Prover to use a large amount of memory.
This is Lemma~\ref{lmm:main_lemma_rnd} below.
We first prove a lemma which captures the property of the random
graph which we need.

\begin{lemma}\label{lem:propertyC}
For a random graph $G$,
with high probability, the following property P holds.
Let \mbox{$U \subseteq V(G)$} with $|U| \le \frac{1}{3} k$ and let
$p$ be any pattern with $|p| \le \frac{1}{3}k$.
Then $p$ is consistent with at least one vertex in $\NB(U)$.
\end{lemma}

\begin{proof}
Fix such a set $U$ and such a pattern $p$. The probability that an arbitrary
vertex $v \notin U$ is in $\NB(U)$ is at least
$2^{-\frac{1}{3}k} = n^{-\frac{1}{3}}$.
The pattern $p$ is consistent with at least $n^\frac{2}{3} - |U|$
vertices outside $U$. The probability that no vertex consistent with $p$
is in $\NB(U)$ is hence at most
\[
\left( 1 - n^{-\frac{1}{3}} \right)^{n^\frac{2}{3} - |U|}
\le e^{- n^\frac{1}{3}}.
\]
We can bound the number of such sets $U$ by $n^{\frac{1}{3}k}  \le n^{\log n}$
and the number of patterns $p$ by $3^k \le n^2$,
so by the union bound
property P fails to hold with probability
at most $2^{-\Omega(n^\frac{1}{3})}$.
\end{proof}

\begin{lemma}\label{lmm:main_lemma_rnd}
Let $G$ be any graph with property P.
Then there is an Adversary strategy in the game on $\clique(G)$
which wins against any Prover
who uses at most $\frac{1}{9}k^2$ memory locations.
\end{lemma}

\begin{proof}
For each index $i \in [ck]$, we will write $p^i$ for the pattern representing
the current information in the Prover's memory about the $i$th vertex.
The Adversary's strategy is to answer queries arbitrarily (say with $0$)
as long as the index $i$ being queried has $|p^i| < \frac{1}{3}k-1$.
If $|p^i| = \frac{1}{3}k-1$, the Adversary privately
\emph{fixes} the $i$th vertex to be some particular vertex $v^i$ of $G$
consistent with $p^i$, and then answers queries to $i$ according to
$v^i$ until, through the Prover forgetting bits, $|p^i|$ falls below
$\frac{1}{3}k$ again, at which point the Adversary considers
the $i$th vertex no longer to be fixed.

If the Adversary is able to guarantee that the set of currently
fixed vertices always forms a clique, then the Prover can never win.
So suppose we are at a point in the game where
the Adversary has to fix a vertex for index $i$, that is, where
the Prover is querying
a bit for  $i$ and $|p^i| = \frac{1}{3}k-1$.
Let $U \subseteq V(G)$ be the set of vertices that the
Adversary currently has fixed. It is enough to show that there is some
vertex consistent with $p^i$ which is connected by an edge in $G$
to every vertex in $U$.
But by the limitation on the size of the Prover's memory, no more than
$\frac{1}{3}k$ vertices can be fixed at any one time. Hence
$|U| \le \frac{1}{3}k$ and the existence
of such a vertex follows from property P.
\end{proof}

\section{Lower bounds for Ramsey graphs} \label{sec:ramsey-graphs}

We prove Theorem \ref{the:ramseylowerbound},
that for any $c$-Ramsey graph $G$ on $n$ vertices,
$L (\ram_G) \ge n^{\Omega( \log n )}$.
As in the previous section we will do this by showing,
in Lemma \ref{lmm:ramsey_strategy} below,  that the
Adversary has a strategy for the game on $\clique(G)$ which forces
the Prover to use a lot of memory.

\begin{definition}
Given sets $A, B \subseteq V(G)$ we define their mutual \emph{density} by
\[
d(A,B) =
\frac{ e(A,B) }
{|A||B|}
\]
where we write $e(A,B)$ for the number of edges in $G$ with one end in $A$ and the other in $B$.
For a single vertex $v$ we will write $d(v, B)$ instead of $d(\{v\}, B)$.
\end{definition}

 Our main tool in our analysis of Ramsey graphs is the statistical property shown in Corollary~\ref{cor:dense_extension} below, which plays a role analogous to that played by
Lemma~\ref{lem:propertyC} for random graphs.
We use the following result proved in~\cite[Case II of Theorem 1]{promel1999non}:

\begin{lemma}[\cite{promel1999non}]%
\label{lmm:dense_pair}
There exists constants $\beta>0$, $\delta>0$
such that if $G$ is a $c$-Ramsey graph, then there is a set $S \subseteq V(G)$ with $|S| \geq n^{\frac{3}{4}}$
such that, for all $A,B \subseteq S$,
if $|A|,|B|\geq |S|^{1-\beta}$ then $\delta \leq d(A,B) \leq 1-\delta$.
\end{lemma}

Now fix a $c$-Ramsey graph $G$. Let $S$, $\beta$ and $\delta$ be as in the above lemma,
and let $m=|S|$. Notice that since our goal is to give an Adversary strategy
 for the formula $\clique(G)$, we will only use the lower bound $\delta \leq d(A,B)$
from the lemma.

\begin{corollary}\label{cor:dense_extension} \sloppypar
 Let $X, Y_{1}, Y_{2}, \ldots, Y_{r} \subseteq S$ be
 such that $|X| \ge rm^{1-\beta}$ and $|Y_{1}|, \ldots, |Y_{r}| \geq m^{1-\beta}$.
Then there exists $v \in X$ such that $d(v,Y_{i}) \geq \delta$ for each $i=1, \ldots, r$.
\end{corollary}

\begin{proof}
For  $i=1,\ldots,r$ let
\[
X_{i} = \{u \in X \;|\; d(u,Y_{i})< \delta  \}.
\]
By Lemma~\ref{lmm:dense_pair}, each $|X_{i}|<m^{1-\beta}$.
Hence $X \setminus \bigcup_{i} X_{i}$ is non-empty and we
can take $v$ to be any vertex in $X \setminus \bigcup_{i} X_{i}$.
\end{proof}

The next lemma implies our main result, Theorem \ref{the:ramseylowerbound}.

\begin{lemma}\label{lmm:ramsey_strategy}
There is a constant $\epsilon>0$, independent of $n$ and $G$,
such that there exists a strategy for the Adversary
in the game on $\clique(G)$ which wins
against any Prover who is limited to $\epsilon^2 k^2$ memory locations.
\end{lemma}

\begin{proof}
Let $\epsilon>0$ be a constant, whose precise value we will fix later.
As in the proof of Lemma \ref{lmm:main_lemma_rnd},
the  Adversary's replies when queried about the $i$th vertex will depend on the
size of $p^i$, the pattern representing the current information
known to the Prover about the $i$th vertex. If $|p^i|<\epsilon k -1$ the Adversary can
reply in a somewhat arbitrary way (see below), but if  $|p^i| = \epsilon k -1$
then the Adversary will fix a value $v^i$ for the $i$th vertex, consistent with $p^i$,
 and will reply according to
$v^i$ until $|p^i|$ falls back below $\epsilon k$, at which point the vertex is no longer fixed.
By the limitation on the Prover's memory, no more than $\epsilon k$
vertices can be fixed simultaneously, which will allow the
Adversary to ensure that the set of currently fixed vertices always forms a clique.

Let $S$, $\beta$ and $\delta$ be as in Lemma \ref{lmm:dense_pair} and
let $m = |S|$. We will need to use Corollary \ref{cor:dense_extension} above
to make sure that the Adversary can find a $v^i$ with suitable
density properties when fixing the $i$th vertex. But here there is a difficulty
which does not arise with the random graph.
Corollary \ref{cor:dense_extension} only works for subsets of the set $S$,
and $S$ may be distributed very non-uniformly over the vertices of $G$. In particular,
through some sequence of querying and forgetting bits for $i$, the
Prover may be able to force the Adversary into a position where the set of
vertices consistent with a small $p^i$ has only a very small intersection
with $S$, so that it is impossible to apply Corollary~\ref{cor:dense_extension}.

Let $\alpha$ be a constant with $0 < \alpha < \beta$, whose precise
value we will fix later. We write $C_p$ for the set of vertices of $G$
consistent with a pattern $p$. We write $P_{\epsilon k}$ for the
set of patterns $p$ with $p \le \epsilon k$. To avoid the problem
in the previous paragraph, we will construct a non-empty set $S^* \subseteq S$
with the property that,  for every $p \in P_{\epsilon k}$, either
\begin{equation*}\label{eq:intersection}
  C_{p}\cap S^* = \emptyset \quad \text{or} \quad |C_{p}\cap S^*|>m^{1-\alpha}.
\end{equation*}
In the second case we will call the pattern $p$ \emph{active}.
The Adversary can then focus on the set $S^*$,
in the sense that he will pretend that his clique is in $S^*$ and
will ignore the vertices outside $S^*$.

We construct $S^*$ in a brute-force way.
We start with $S_{0} =S$ and define a sequence of subsets $S_0, S_1, \ldots$
where each $S_{t+1} = S_{t}\setminus{C_{p}}$ for the lexicographically first
$p \in P_{\epsilon k}$ for which $0< |S_{t}\cap{C_{p}}| \leq m^{1-\alpha}$,
if any such $p$ exists. We stop as soon as there is no such $p$, and let
 $S^*$ be the final subset in the sequence.
To show that $S^*$ is non-empty, notice that at
each step at most $m^{1-\alpha}$ elements are removed.
Furthermore there are at most $|P_{\epsilon k}|$ steps, since a set of vertices $C_{p}$ may be removed at most once. Recall that $n = 2^k$ and $m \ge n^\frac{3}{4}$. We have
\begin{equation*}
|P_{\epsilon k}| = \sum^{\epsilon k}_{i=0} 2^{i} \binom{k}{i}
\leq  \epsilon k \cdot 2^{\epsilon k} \binom{k}{ \epsilon k}
\leq \epsilon k  \cdot  n^{\epsilon} n^{H(\epsilon)},
\end{equation*}
where $H(x)$ is the binary entropy function $- x \log x - (1-x)\log (1-x)$,
and we are using the estimate $\binom{k}{\epsilon k}\leq 2^{k H(\epsilon )}$ which holds for $0<\epsilon <1$.
Then
\[
|S^*| \ge |S| - |P_{\epsilon k}|  \cdot  m^{1-\alpha}
\ge n^\frac{3}{4} - \epsilon k  \cdot n^{\epsilon + H(\epsilon)} n^{ \frac{3}{4} ( 1 - \alpha )},
\]
so, for large $n$,  $S^*$ is non-empty as long as we choose $\alpha$ and $\epsilon$
satisfying
\begin{equation}\label{eq:parameter_condition1} \tag{$\star$}
 \tfrac{3}{4}\alpha > \epsilon + H(\epsilon).
\end{equation}
Notice that if $S^*$ is non-empty then in fact $|S^*| > m^{1-\alpha}$,
since $S^*$ must intersect at least the set $C_p$ where $p$
is the empty pattern.

We can now give the details of the Adversary's strategy.
The Adversary maintains the following three conditions,
which in particular guarantee that the Prover will never win.
\begin{enumerate}
\item
For each index $i$, if $|{p^i}| < \epsilon k$ then $p^i$ is active, that is, $C_{p^i} \cap S^* \neq \emptyset$.
\item
For each index $i$, if $|{p^i}| \ge \epsilon k$ then the $i$th vertex is fixed to some
$v^i \in C_{p^i} \cap S^*$; furthermore the set $U$ of currently fixed vertices $v^j$ forms a clique.
\item \label{cond:partial_clique_dense}
For every active $p \in P_{\epsilon k}$ and every $U' \subseteq U$, we have
\begin{equation*}
  |C_{p} \cap S^* \cap \NB(U')| \geq |C_{p} \cap S^*|\cdot \delta^{|U'|}.
\end{equation*}
\end{enumerate}
These are true at the start of the game, because no vertices are fixed
and each $p^i$ is the empty pattern.

Suppose that, at a turn in the game, the Prover queries a bit for
an index $i$ for which he currently has information $p^i$.
If $|p^i|<\epsilon k -1$, then by condition 1 there is at least one vertex $v$ in
$C_{p^i} \cap S^*$. The Adversary chooses an arbitrary such $v$ and replies according to the bit of $v$.
If $|p^i| \ge \epsilon k$, then a vertex $v^i \in C_{p^i}$ is already fixed,
and the Adversary replies according to the bit of $v^i$.

If $|p^i|=\epsilon k -1$,  then the Adversary must fix a vertex $v^i$ for $i$
in a way that satisfies conditions 2 and 3.
To preserve condition~2, $v^i$ must be connected to every vertex in the set $U$
of currently fixed vertices.
To preserve condition~3, it is enough to choose $v^i$ such that
\[
d(v^i,C_p \cap S^* \cap \NB(U')) \geq
|C_{p} \cap S^* \cap \NB(U')|  \cdot \delta
\]
for every active $p$ in $P_{\epsilon k}$ and every $U' \subseteq U$.
To find such a $v^i$ we will
apply Corollary~\ref{cor:dense_extension},
with one set $Y$ for each pair of a suitable $p$ and $U'$.
We put
\begin{align*}
X & = C_{p^i} \cap \NB(U) \cap S^*\\
Y_{(p,U')} & =C_{p} \cap \NB(U') \cap S^*
\text{\ for each active $p \in P_{\epsilon k}$ and each $U'\subseteq U$}\\
r & = | \{ \text{pairs } (p, U') \}| \le | P_{\epsilon k}| \cdot 2^{|U|}.
\end{align*}
We know $|U| \le \epsilon k$. By condition 1 we know $p^i$ is active, hence
$|C_{p^i} \cap S^*| > m^{1-\alpha}$. So by condition 3 we have
\[
|X| \ge m^{1-\alpha}  \delta^{\epsilon k} = m^{1 - \alpha + \frac{4}{3}\epsilon\log \delta}.
\]
For similar reasons we have the same lower bound on the size of each $Y_{(p,U')}$.
Furthermore
\[
r \le 2^{\epsilon k} \cdot \epsilon k \cdot n^{\epsilon + H(\epsilon)}
= \epsilon k \cdot n^{2 \epsilon + H(\epsilon)}
= \epsilon k \cdot m^{\frac{8}{3} \epsilon + \frac{4}{3} H(\epsilon)}.
\]
To apply Corollary~\ref{cor:dense_extension} we need to satisfy
$|X|\geq r m^{1-\beta}$ and $|Y_{(p,U')}| \geq m^{1-\beta}$. Both conditions are implied by the inequality
\begin{equation}\label{eq:parameter_cond2} \tag{$\dagger$}
\beta - \alpha > \tfrac{8}{3}\epsilon + \tfrac{4}{3}H(\epsilon) - \tfrac{4}{3}\epsilon\log \delta.
\end{equation}

We can now fix values for the constants $\alpha$ and $\epsilon$ to
satisfy the inequalities (\ref{eq:parameter_condition1}) and (\ref{eq:parameter_cond2}).
Since $H(\epsilon)$ goes to zero as $\epsilon$ goes to zero, we can make the right hand sides of
(\ref{eq:parameter_condition1}) and (\ref{eq:parameter_cond2}) arbitrary small by setting $\epsilon$ to be a small constant. We then set $\alpha$ appropriately.

Finally, it is straightforward to check that if the Prover forgets a bit for an index~$i$, then the
three conditions are preserved.
\end{proof}

\section*{Acknowledgements}
Part of this work was done while Lauria was at the Institute of Mathematics of the
Academy of Sciences of the Czech Republic, supported by the Eduard \v{C}ech Center.
Lauria, \Pudlak and Thapen did part of this research at the
Isaac Newton Institute for the Mathematical Sciences, where \Pudlak and Thapen
were visiting fellows in the programme \textit{Semantics and Syntax}.
 \Pudlak and Thapen were also supported by grant IAA100190902 of GA AV \v{C}R,
and by Center of Excellence CE-ITI under grant P202/12/G061 of GA \v{C}R and RVO: 67985840.
Lauria was also supported by the European Research Council under the
European Union's Seventh Framework Programme \mbox{(FP7/2007--2013) /} ERC
grant agreement no~279611


\bibliography{theoryofcomputing}
\bibliographystyle{abbrv}

\end{document}